\documentclass[12pt]{article}

\usepackage[margin=1in]{geometry}
\usepackage{amsthm}
\usepackage{amssymb}
\usepackage{amsmath}
\usepackage{graphicx}
\usepackage{mathrsfs}
\usepackage{url}
\usepackage{color}
\usepackage{framed}
\usepackage[table]{xcolor}
\usepackage[ruled,vlined]{algorithm2e}
\usepackage{mathtools}

\newtheorem{theorem}{Theorem}

\newtheorem{lemma}[theorem]{Lemma}
\newtheorem{proposition}[theorem]{Proposition}

\theoremstyle{definition}

\newtheorem{definition}[theorem]{Definition}

\title{Derandomized compressed sensing\\with nonuniform guarantees for $\ell_1$ recovery}

\author{Charles Clum \and Dustin~G.~Mixon}
\date{}

\begin{document}
\maketitle

\begin{abstract}
We extend the techniques of H\"{u}gel, Rauhut and Strohmer~\cite{HugelRS:14} to show that for every $\delta\in(0,1]$, there exists an explicit random $m\times N$ partial Fourier matrix $A$ with $m=s\operatorname{polylog}(N/\epsilon)$ and entropy $s^\delta\operatorname{polylog}(N/\epsilon)$ such that for every $s$-sparse signal $x\in\mathbb{C}^N$, there exists an event of probability at least $1-\epsilon$ over which $x$ is the unique minimizer of $\|z\|_1$ subject to $Az=Ax$.
The bulk of our analysis uses tools from decoupling to estimate the extreme singular values of the submatrix of $A$ whose columns correspond to the support of $x$.
\end{abstract}

\section{Introduction}

A vector $x\in\mathbb{C}^N$ is said to be \textbf{$s$-sparse} if it has at most $s$ nonzero entries.
Natural images are well-approximated by sparse vectors in a wavelet domain, and this feature is exploited by JPEG2000 image compression~\cite{TaubmanM:12}.
In 2006, Cand\`{e}s, Romberg and Tao~\cite{CandesRT:06} and Donoho~\cite{Donoho:06} discovered that sparsity could also be exploited for \textbf{compressed sensing}.
One popular formulation of compressed sensing is to find a sensing matrix $A\in\mathbb{C}^{m\times N}$ such that every $s$-sparse vector $x\in\mathbb{C}^N$ with $s\leq m/\operatorname{polylog} N$ can be efficiently reconstructed from the linear data $y=Ax$ by solving the convex program
\begin{equation}
\label{eq.L1min}
\text{minimize}
\quad
\|z\|_1
\quad
\text{subject to}
\quad
Az=y.
\end{equation}
To enjoy this $\ell_1$-recovery property, it suffices for $A$ to act as a near-isometry over the set of $s$-sparse vectors~\cite{Candes:08,CaiZ:13}:
\begin{equation}
\label{eq.rip}
\frac{1}{2}\|x-y\|_2^2
\leq \|Ax-Ay\|_2^2
\leq\frac{3}{2}\|x-y\|_2^2
\qquad
\text{for every $s$-sparse $x,y \in \mathbb{C}^N$}.
\end{equation}
We refer to such $A$ as \textbf{$s$-restricted isometries}.
Equivalently, every submatrix $A_T$ that is comprised of $2s$ columns from $A$ has singular values $\sigma(A_T)\subseteq[\sqrt{1/2},\sqrt{3/2}]$.
Since random matrices exhibit predictable extreme singular values~\cite{Tao:12,Tropp:15,Vershynin:18}, it comes as no surprise that many distributions of random matrices $A\in\mathbb{C}^{m\times N}$ are known to be $s$-restricted isometries with high probability provided $m\geq s\operatorname{polylog}N$, e.g.,~\cite{CandesT:06,BaraniukDDW:08,MendelsonPT:09,KrahmerMR:14}.
Unfortunately, testing \eqref{eq.rip} is $\mathsf{NP}$-hard~\cite{BandeiraDMS:13,TillmannP:13}, and it is even hard for matrices with independent subgaussian entries, assuming the hardness of finding planted dense subgraphs~\cite{WangBP:16}.

In 2007, Tao~\cite{Tao:07} posed the problem of finding explicit $s$-restricted isometries $A\in\mathbb{C}^{m\times N}$ with $N^\epsilon\leq m\leq(1-\epsilon)N$ and $m=s\operatorname{polylog}N$.
One may view this as an instance of Avi Wigderson's \textit{hay in a haystack} problem~\cite{BandeiraFMM:16}.
To be clear, we say a sequence $\{A_N\}$ of $m(N)\times N$ matrices with $N\to\infty$ is \textbf{explicit} if there exists an algorithm that on input $N$ produces $A_N$ in time that is polynomial in $N$.
For example, we currently know of several explicit sequences of matrices $A$ with unit-norm columns $\{a_i\}_{i\in[N]}$ and minimum \textbf{coherence}:
\[
\max_{\substack{i,j\in[N]\\ i\neq j}}|\langle a_i,a_j\rangle|.
\]
See~\cite{FickusM:15} for a survey.
Since the columns of such matrices are nearly orthonormal, they are intuitively reasonable choices to ensure $\sigma(A_T)\subseteq[\sqrt{1/2},\sqrt{3/2}]$.
One may leverage the Gershgorin circle theorem to produce such an estimate~\cite{ApplebaumHSC:09,DeVore:07,BandeiraFMW:13}, but this will only guarantee \eqref{eq.rip} for $s\leq m^{1/2}/\operatorname{polylog}N$.
In fact, this estimate is essentially tight since there exist $m\times N$ matrices of minimum coherence with $\Theta(\sqrt{m})$ linearly dependent columns~\cite{FickusMT:12,JasperMF:13}.
As an alternative to Gershgorin, Bourgain et al.~\cite{BourgainDFKK:11,BourgainDFKK:11b} introduced the so-called \textit{flat RIP} estimate to demonstrate that certain explicit $m\times N$ matrices with $m=\Theta(N^{1-\epsilon})$ are $s$-restricted isometries for $s=O(m^{1/2+\epsilon})$, where $\epsilon=10^{-16}$; see~\cite{Mixon:15} for an expository treatment.
It was conjectured in~\cite{BandeiraFMW:13} that the Paley equiangular tight frames~\cite{Renes:07} are restricted isometries for even larger values of $\epsilon$, and the flat RIP estimate can be used to prove this, conditional on existing conjectures on cancellations in the Legendre symbol~\cite{BandeiraMM:17}.

While is it difficult to obtain explicit $s$-restricted isometries for $s=m/\operatorname{polylog}N$, there have been two approaches to make partial progress: \textit{random signals} and \textit{derandomized matrices}.
The random signals approach explains a certain observation:
While low-coherence $m\times N$ sensing matrices may not determine \textit{every} $s$-sparse signal with $s=m/\operatorname{polylog}N$, they do determine \textit{most} of these signals.
In fact, even for the $m\times N$ matrices $A$ of minimum coherence with $\Theta(\sqrt{m})$ linearly dependent columns (indexed by $T$, say), while $y=Ax$ fails to uniquely determine any $x$ with support containing $T$, it empirically holds that random $s$-sparse vectors $x$ can be reconstructed from $y=Ax$ by solving \eqref{eq.L1min}.
This behavior appears to exhibit a phase transition~\cite{MonajemiJGD:13}, and Tropp proved this behavior up to logarithmic factors in~\cite{Tropp:08}; see also the precise asymptotic estimates conjectured by Haikin, Zamir and Gavish~\cite{HaikinZG:17} and recent progress in~\cite{MagsinoMP:19}.

As another approach, one may seek explicit random matrices that are $s$-restricted isometries for $s=m/\operatorname{polylog}N$ with high probability, but with as little entropy as possible; here, we say a sequence $\{A_N\}$ of $m(N)\times N$ random matrices is explicit if there exists an algorithm that on input $N$ produces $A_N$ in time that is polynomial in $N$, assuming access to a random variable that is uniformly distributed over $[k]:=\{1,\ldots,k\}$ for any desired $k\in\mathbb{N}$.
Given a discrete random variable $X$ that takes values in $\mathcal{X}$, the \textbf{entropy} $H(X)$ of $X$ is defined by
\[
H(X)
:=-\sum_{x\in\mathcal{X}}\mathbb{P}\{X=x\}\log_2\mathbb{P}\{X=x\}.
\]
For example, the uniform distribution over $[2^H]$ has entropy $H$, meaning it takes $H$ independent tosses of a fair coin to simulate this distribution.
One popular random matrix in compressed sensing draws independent entries uniformly over $\{\pm m^{-1/2}\}$~\cite{CandesT:06,BaraniukDDW:08,DuarteDTLSKB:08,MendelsonPT:09}, which has entropy $H=mN=sN\operatorname{polylog}N$.
One may use the Legendre symbol to derandomize this matrix to require only $H=s\operatorname{polylog}N$ random bits~\cite{BandeiraFMM:16}.
Alternatively, one may draw $m$ rows uniformly from the $N\times N$ discrete Fourier transform to get $H=m\log_2 N=s\operatorname{polylog}N$~\cite{CandesT:06,RudelsonV:08,CheraghchiGV:13,Bourgain:14,HavivR:17}.
Any choice of Johnson--Lindenstrauss projection~\cite{JohnsonL:86} with $m=s\operatorname{polylog}N$ is an $s$-restricted isometry with high probability~\cite{BaraniukDDW:08}, but these random matrices inherently require $H=\Omega(m)$~\cite{BandeiraFMM:16}.
To date, it is an open problem to find explicit random $m\times N$ matrices with $N^\epsilon\leq m\leq(1-\epsilon)N$ and $H\ll s$ that are $s$-restricted isometries for $s=m/\operatorname{polylog}N$ with high probability.

There is another meaningful way to treat the compressed sensing problem:
Show that the distribution of a given random $m\times N$ matrix $A$ has the property that for every $s$-sparse $x\in\mathbb{C}^N$, there exists a high-probability event $\mathcal{E}(x)$ over which $x$ can be reconstructed from $y=Ax$ by solving \eqref{eq.L1min}.
This \textbf{nonuniform} setting was originally studied by Cand\`{e}s, Romberg and Tao in~\cite{CandesRT:06}, and later used to define the Donoho--Tanner phase transition~\cite{DonohoT:09,AmelunxenLMT:14}.
For applications, the nonuniform setting assumes that a fresh copy of $A$ is drawn every time a signal $x$ is to be sensed as $y=Ax$, and then both $A$ and $y$ are passed to the optimizer to solve \eqref{eq.L1min}, which is guaranteed to recover $x$ at least $\inf_x\mathbb{P}(\mathcal{E}(x))$ of the time.
Notice that if an explicit random matrix is an $s$-restricted isometry in the high-probability event $\mathcal{E}$, then it already enjoys a such guarantee with $\mathcal{E}(x)=\mathcal{E}$ for every $s$-sparse $x\in\mathbb{C}^N$.
As such, it is natural to seek a nonuniform guarantee for an explicit random matrix with entropy $H\ll s$.

Let $G\colon\{0,1\}^*\to\{0,1\}^*$ denote a \textbf{pseudorandom number generator} with stretching parameter $L\in\mathbb{N}$, that is, a deterministic function that is computable in polynomial time such that (i) $G$ maps strings of length $n$ to strings of length $n^L$, and (ii) for every function $D\colon\{0,1\}^*\to\{0,1\}$ that is computable in probabilistic polynomial time and every $k\in\mathbb{N}$, there exists $n_0\in\mathbb{N}$ such that for every $n\geq n_0$, it holds that
\begin{equation*}
|\mathbb{P}\{D(G(V))=1\}-\mathbb{P}\{D(U)=1\}|< n^{-k},
\end{equation*}
where $U$ and $V$ are uniformly distributed over $\{0,1\}^{n^L}$ and $\{0,1\}^{n}$, respectively.
In words, $G$ stretches $n$ random bits into $n^L$ bits that are computationally indistinguishable from true randomness.
While pseudorandom number generators are not known to exist, their existence is a fundamental assumption in modern cryptography~\cite{Goldreich:01}.
Take any choice of explicit $m\times N$ random matrices $A(U)$ with $m=s\operatorname{polylog}N$ and entropy $s\operatorname{polylog}N$ that are $s$-restricted isometries with high probability, and consider the pseudorandom counterpart $A(G(V))$ with entropy $s^{1/L}\operatorname{polylog}N$.
If there were an $s$-sparse $x\in\mathbb{C}^N$ that failed to typically equal the unique minimizer of $\|z\|_1$ subject to $A(G(V))z=A(G(V))x$, then we could use this $x$ to detect the difference between $U$ and $G(V)$.
As such, we expect that for every $L\in\mathbb{N}$, there exist explicit $m\times N$ random matrices with $m=s\operatorname{polylog}N$ and entropy $s^{1/L}\operatorname{polylog}N$ that enjoy a nonuniform $\ell_1$-recovery guarantee.


In 2014, H\"{u}gel, Rauhut and Strohmer~\cite{HugelRS:14} applied tools from decoupling to show that certain random matrices that arise in remote sensing applications enjoy a nonuniform $\ell_1$-recovery guarantee.
One may directly apply their techniques to obtain an explicit random $m\times N$ partial Fourier matrix $A$ with entropy $\Theta(s^{1/2}\log N\log(N/\epsilon))$ such that each $s$-sparse vector $x\in\mathbb{C}^N$ can be recovered by $\ell_1$ minimization with probability at least $1-\epsilon$.
Explicitly, $A$ is the submatrix of the $N\times N$ discrete Fourier transform with rows indexed by $\{b_i+b_j:i,j\in[n]\}$, where $n=\Theta(s^{1/2}\log(N/\epsilon))$ and $b_1,\ldots, b_n$ are independent random variables with uniform distribution over $\mathbb{Z}/N\mathbb{Z}$.
In this paper, we generalize this construction to allow for row indices of the form $\{b_{i_1}+\cdots+b_{i_L}:i_1,\ldots,i_L\in[n]\}$.
Our main result, found in the next section, is that for each $L$, one may take $n=\Theta_L(s^{1/L}\log^4(N/\epsilon))$ to obtain a nonuniform $\ell_1$-recovery guarantee.
In a sense, this confirms our prediction from the previous paragraph, but it does not require the existence of pseudorandom number generators.
Our proof hinges on a different decoupling result (Proposition~\ref{decouplingtheorem}) that reduces our key spectral norm estimate to an iterative application of the matrix Bernstein inequality (Proposition~\ref{matrixbernstein}); see Section~3.
In Section~4, we provide a simplified treatment of the moment method used in~\cite{HugelRS:14} to obtain an approximate dual certificate, though generalized for our purposes.
Hopefully, similar ideas can be used to produce explicit random $m\times N$ matrices with $m=s\operatorname{polylog}N$ and entropy $H\ll s$ that are $s$-restricted isometries with high probability.
Also, we note that Iwen~\cite{Iwen:14} identified explicit random $m\times N$ matrices with $m=\Theta(s\log^2N)$ and entropy $H=\Theta(\log^2s)$ for which a specialized algorithm enjoys a nonuniform recovery guarantee, and it would be interesting if this level of derandomization could also be achieved with $\ell_1$ recovery.

\section{Main result}

Throughout, we take $e_N\colon\mathbb{R}\to\mathbb{C}$ defined by $e_N(x):=e^{2\pi i x/N}$.
We will use the following random matrix as a compressed sensing matrix:

\begin{definition}
Let $N$ be prime and put $m = n^{L}$ for some $n,L\in \mathbb{N}$.
Given independent, uniform random variables $b_{1},\ldots,b_{n}$ over $[N]$, then the \textbf{Minkowski partial Fourier matrix} is the random $m\times N$ matrix $A$ with rows indexed by $[n]^{L}$, and whose entry at $(\mathcal{I},j)=(i_{1},\ldots,i_{L},j)\in [n]^{L}\times [N]$ is given by $A_{\mathcal{I},j}:= n^{-L/2}\cdot e_N((b_{i_1} + \cdots + b_{i_L}) j)$.
\end{definition}

We take $N$ to be prime for convenience; we suspect that our results also hold when $N$ is not prime, but the proofs would be more complicated.
To perform compressed sensing, we sense with the random matrix $A$ to obtain data $y$ and then we solve the following program:
\begin{equation}
\label{noisyell1minimization}
\text{minimize}
\quad
\|z\|_1
\quad
\text{subject to}
\quad
\|Az-y\|_2\leq \eta.
\end{equation}
Our main result states that if $x\in\mathbb{C}^N$ is nearly sparse, $n=s^{1/L}\operatorname{polylog}(N/\epsilon)$, and the noisy data $y=Ax+e$ satisfies $\|e\|_2\leq\eta$, then the minimizer of \eqref{noisyell1minimization} is a good approximation of $x$ with probability at least $1-\epsilon$:

\begin{theorem}[main result]
\label{desiredderandomizationtheorem}
Fix $L\in\mathbb{N}$.
There exists $C_L>0$ depending only on $L$ such that the following holds.
Given any prime $N$, positive integer $s\leq N$, and $\epsilon\in(0,1)$, select any integer $n\geq C_{L}s^{1/L}\log^{4}(N/\epsilon)$.
Then the Minkowski partial Fourier matrix $A$ with parameters $(N,n,L)$ has the property that for any fixed signal $x\in \mathbb{C}^{N}$ and noise $e\in\mathbb{C}^{n^L}$ with $\|e\|_{2}\leq \eta$, then given random data $y = Ax + e$, the minimizer $\hat{x}$ of~\eqref{noisyell1minimization} satisfies the estimate
\[
\|\hat{x}-x\|_2
\leq 25\cdot(\sqrt{s}\cdot\eta + \|x-x_s\|_{1})
\]
with probability at least $1 - \epsilon$.  
\end{theorem}

To prove this result, we construct an approximate dual certificate that satisfies the hypotheses of the following proposition.
Here and throughout, $x_s$ denotes (any of) the best $s$-term approximation(s) of $x$.

\begin{proposition}[Theorem~3.1 in~\cite{HugelRS:14}, cf.\ Theorem~4.33 in~\cite{FoucartR:13}]
\label{prop.approx dual cert}
Fix a signal $x\in \mathbb{C}^{N}$ and a measurement matrix $A\in\mathbb{C}^{m\times N}$ with unit-norm columns.
Fix $s\in[N]$ and any $T\subseteq[N]$ of size $s$ such that $\|x-x_T\|_1=\|x-x_s\|_1$.
Suppose
\begin{equation}
\label{eq.conditioning}
\|A_T^*A_T-I\|_{2\to2}
\leq1/2
\end{equation}
and that there exists $v\in\mathbb{C}^m$ such that $u:=A^*v$ satisfies
\begin{equation}
\label{eq.dual cert}
u_T=\operatorname{sgn}(x)_T,
\qquad
\|u_{T^c}\|_\infty\leq 1/2,
\qquad
\|v\|_2\leq \sqrt{2s}.
\end{equation}
Then for every $e\in\mathbb{C}^{m}$ with $\|e\|_{2}\leq \eta$, given random data $y = Ax + e$, the minimizer $\hat{x}$ of~\eqref{noisyell1minimization} satisfies the estimate
\[
\|\hat{x}-x\|_2
\leq 25\cdot(\sqrt{s}\cdot\eta + \|x-x_s\|_{1}).
\]
\end{proposition}

When applying Proposition~\ref{prop.approx dual cert}, we will make use of a few intermediate lemmas.
For example, the columns of $A$ have low coherence with high probability:

\begin{lemma}
\label{lem.coherence}
Fix $L\in\mathbb{N}$.
There exists $C^{(1)}_L>0$ depending only on $L$ such that the following holds.
Given any $N,n\in\mathbb{N}$, let $\{a_i\}_{i\in[N]}$ denote the column vectors of the Minkowski partial Fourier matrix $A$ with parameters $(N,n,L)$.
Then
\[
\max_{\substack{i,j\in[N]\\i\neq j}}|\langle a_i,a_j\rangle|
\leq C^{(1)}_L\cdot n^{-L/2} \cdot \log^{L/2}(N/\epsilon)
\]
with probability at least $1-\frac{\epsilon}{3}$.
\end{lemma}

The proof of Lemma~\ref{lem.coherence} follows from an application of the complex Hoeffding inequality:

\begin{proposition}[complex Hoeffding]
\label{prop.complex Hoeffding}
Suppose $X_{1},\ldots, X_{n}$ are independent complex random variables with
\[
\mathbb{E}X_i=0
\qquad
\text{and}
\qquad
|X_i|\leq K
\qquad
\text{almost surely}.
\]
Then for every $t\geq0$, it holds that
\[
\mathbb{P}\bigg\{\bigg|\sum_{i=1}^{n}X_{i}\bigg| > t\bigg\}
\leq 4\operatorname{exp}\bigg(-\frac{t^{2}}{2nK^2}\bigg).
\]
\end{proposition}

\begin{proof}
Write $X_i=A_i+\sqrt{-1}\cdot B_i$.
Then $|\sum_iX_{i}|^2=|\sum_iA_{i}|^2+|\sum_iB_{i}|^2$, and so
\[
\mathbb{P}\bigg\{\bigg|\sum_{i=1}^{n}X_{i}\bigg| > t\bigg\}
\leq \mathbb{P}\bigg\{\bigg|\sum_{i=1}^{n}A_{i}\bigg| > \frac{t}{\sqrt{2}}\bigg\}+\mathbb{P}\bigg\{\bigg|\sum_{i=1}^{n}B_{i}\bigg| > \frac{t}{\sqrt{2}}\bigg\}
\leq 4\operatorname{exp}\bigg(-\frac{t^2}{2nK^2}\bigg),
\]
where the last step applies Hoeffding's inequality, namely, Theorem~2 in~\cite{Hoeffding:63}.
\end{proof}

\begin{proof}[Proof of Lemma~\ref{lem.coherence}]
For each $j\in[N]$, consider $v_j\in\mathbb{C}^n$ whose $i$-th entry is $e_N(b_ij)$, and observe that $a_j=n^{-L/2}\cdot v_j^{\otimes L}$.
Apply the union bound and Proposition~\ref{prop.complex Hoeffding} to get
\[
\mathbb{P}\bigg\{\max_{\substack{j,k\in[N]\\j\neq k}}|\langle v_j,v_k\rangle|>t\bigg\}
\leq \sum_{\substack{j,k\in[N]\\j\neq k}}\mathbb{P}\bigg\{ \bigg|\sum_{i=1}^n e_N(b_i(j-k))\bigg| > t\bigg\}
\leq 4N^2e^{-t^2/(2n)}
\leq \frac{\epsilon}{3},
\]
where the last step takes $t:=(C^{(1)}_L)^{1/L}\cdot n^{1/2}\cdot \log^{1/2}(N/\epsilon)$.
It follows that
\[
\max_{\substack{j,k\in[N]\\j\neq k}}|\langle a_j,a_k\rangle|
=n^{-L}\max_{\substack{j,k\in[N]\\j\neq k}}|\langle v_j,v_k\rangle|^L
\leq C^{(1)}_L\cdot n^{-L/2} \cdot \log^{L/2}(N/\epsilon)
\]
with probability at least $1-\frac{\epsilon}{3}$, as desired.
\end{proof}

The proofs of the following lemmas can be found in the Sections~3 and~4:

\begin{lemma}
\label{conditioningtheorem}
Fix $L\in\mathbb{N}$.
There exists $C^{(2)}_L>0$ depending only on $L$ such that the following holds.
Given any prime $N$, positive integer $s\leq N$, and $\epsilon\in(0,1)$, select any integer $n\geq C^{(2)}_{L}s^{1/L}\log^{3}(N/\epsilon)$.
Then for every $T\subseteq [N]$ of size $|T| = s$, the Minkowski partial Fourier matrix $A$ with parameters $(N,n,L)$ satisfies $\|A_{T}^{*}A_{T} - I\|_{2\rightarrow 2} \leq \frac{1}{e}$ with probability at least $1-\frac{\epsilon}{3}$.
\end{lemma}

\begin{lemma}
\label{lem.moment hell}
Fix $N,n,L,s\in\mathbb{N}$ such that $n\geq 2s^{1/L}$, and take any $T\subseteq[N]$ of size $s$ and any $z\in\mathbb{C}^T$ with $\|z\|_\infty\leq1$.
Let $\{a_i\}_{i\in[N]}$ denote the column vectors of the Minkowski partial Fourier matrix $A$ with parameters $(N,n,L)$.
Take any $u\in T^c$, $k\in\mathbb{N}$ and $p\geq 2$, and put $\eta:=(k+1)Lp$.
Then
\[
\mathbb{E}|a_u^*A_T(I-A_T^*A_T)^k z|^{p}
\leq 2\eta^{2\eta}(s^{1/L}n^{-1})^{\eta/2}.
\]
\end{lemma}

\begin{proof}[Proof of Theorem~\ref{desiredderandomizationtheorem}]
We will restrict to an event of the form $\mathcal{E}=\mathcal{E}_1\cap\mathcal{E}_2\cap\mathcal{E}_3$ such that the hypotheses of Proposition~\ref{prop.approx dual cert} are satisfied over $\mathcal{E}$ and each $\mathcal{E}_i$ has probability at least $1-\frac{\epsilon}{3}$.
By Lemma~\ref{conditioningtheorem}, we may take
\[
\mathcal{E}_1:=\{\|A_{T}^{*}A_{T} - I\|_{2\rightarrow 2} \leq 1/e\}.
\]
Then \eqref{eq.conditioning} is satisfied over $\mathcal{E}_1\supseteq\mathcal{E}$.
For \eqref{eq.dual cert}, we take $v:=(A_T^*)^\dagger\operatorname{sgn}(x)_T$.
Since $A_T$ has trivial kernel over $\mathcal{E}_1$, we may write $(A_T^*)^\dagger=A_T(A_T^*A_T)^{-1}$, and so
\[
u_T
=(A^*v)_T
=A_T^*v
=A_T^*(A_T^*)^\dagger\operatorname{sgn}(x)_T
=A_T^*A_T(A_T^*A_T)^{-1}\operatorname{sgn}(x)_T
=\operatorname{sgn}(x)_T.
\]
Next, the nonzero singular values of any $B$ and $B^\dagger$ are reciprocal to each other, and so
\begin{align*}
\|v\|_2
=\|(A_T^*)^\dagger\operatorname{sgn}(x)_T\|_2
\leq\|(A_T^*)^\dagger\|_{2\to2}\cdot\|\operatorname{sgn}(x)_T\|_2
\leq (1-1/e)^{-1/2}\cdot s^{1/2}
\leq \sqrt{2s}.
\end{align*}
As such, for \eqref{eq.dual cert}, it remains to verify that $\|u_{T^c}\|_\infty\leq 1/2$.
We will define $\mathcal{E}_2$ and $\mathcal{E}_3$ in such a way that Lemmas~\ref{lem.coherence} and~\ref{lem.moment hell} imply this bound.
Since $\|I-A_T^*A_T\|_{2\to2}<1$, we may write $(A_T^*A_T)^{-1}=\sum_{k=0}^\infty (I-A_T^*A_T)^k$.
For any $\omega\in\mathbb{N}$, the triangle inequality then gives
\begin{align}
\|u_{T^c}\|_\infty
=\|A_{T^c}^*v\|_\infty
\nonumber
&=\|A_{T^c}^*A_T(A_T^*A_T)^{-1}\operatorname{sgn}(x)_T\|_\infty\\
\label{eq.bound with moments}
&\leq \bigg\|A_{T^c}^*A_T\sum_{k=0}^{\omega-1} (I-A_T^*A_T)^k\operatorname{sgn}(x)_T\bigg\|_\infty\\
\label{eq.bound with coherence}
&\qquad + \bigg\|A_{T^c}^*A_T\sum_{k=\omega}^\infty (I-A_T^*A_T)^k\operatorname{sgn}(x)_T\bigg\|_\infty.
\end{align}
Given $B\in\mathbb{C}^{a\times b}$, denote $\|B\|_{\infty}:=\max_{i\in[a],j\in[b]}|B_{ij}|$.
Observe that $\|B\|_{\infty\to\infty}\leq b\|B\|_\infty$, and recall that $\|B\|_{\infty\to\infty}\leq\sqrt{b}\|B\|_{2\to2}$.
We apply these estimates to obtain
\begin{align*}
\eqref{eq.bound with coherence}
&\leq\|A_{T^c}^*A_T\|_{\infty\to\infty}\cdot\bigg\|\sum_{k=\omega}^\infty (I-A_T^*A_T)^k\bigg\|_{\infty\to\infty}\\
&\leq s\|A_{T^c}^*A_T\|_{\infty}\cdot \sqrt{s}\sum_{k=\omega}^\infty \|I-A_T^*A_T\|_{2\to2}^k
\leq s^{3/2}\cdot\max_{\substack{i,j\in[N]\\i\neq j}}|\langle a_i,a_j\rangle| \cdot \sum_{k=\omega}^\infty e^{-k}.
\end{align*}
Next, by Lemma~\ref{lem.coherence}, we may take
\[
\mathcal{E}_2
:=\bigg\{
\max_{\substack{i,j\in[N]\\i\neq j}}|\langle a_i,a_j\rangle|
\leq C^{(1)}_L\cdot n^{-L/2} \cdot \log^{L/2}(N/\epsilon)
\bigg\}.
\]
Then since $n\geq (C^{(1)}_L)^{2/L}\cdot s^{1/L} \cdot \log(N/\epsilon)$ by assumption, we have
\[
\eqref{eq.bound with coherence}
\leq s^{3/2}\cdot\max_{\substack{i,j\in[N]\\i\neq j}}|\langle a_i,a_j\rangle| \cdot \sum_{k=\omega}^\infty e^{-k}
\leq s^{3/2} \cdot s^{-1/2} \cdot \sum_{k=\omega}^\infty e^{-k}
=s(1-1/e)^{-1}e^{-\omega}
\leq 1/4
\]
over $\mathcal{E}_2\supseteq\mathcal{E}$, where the last step follows from selecting $\omega:=\lceil 2\log N\rceil$, say.
It remains to establish $\eqref{eq.bound with moments}\leq 1/4$.
To this end, we define
\[
\mathcal{E}_3
:=\{\eqref{eq.bound with moments} \leq 1/4\},
\]
and we will use Lemma~\ref{lem.moment hell} to prove $\mathbb{P}(\mathcal{E}_3^c)\leq \epsilon/3$ by the moment method.
For every choice of $\beta_0,\ldots,\beta_{\omega-1}>0$ satisfying $\sum_{k=0}^{\omega-1}\beta_k\leq 1/4$ and $p_0,\ldots,p_{\omega-1}\geq2$, the union bound gives
\begin{align}
\mathbb{P}(\mathcal{E}_3^c)
\nonumber
&=\mathbb{P}\bigg\{\max_{u\in T^c}\bigg|a_u^*A_T\sum_{k=0}^{\omega-1} (I-A_T^*A_T)^k\operatorname{sgn}(x)_T\bigg| > 1/4\bigg\}\\
\nonumber
&\leq\sum_{u\in T^c}\mathbb{P}\bigg\{\sum_{k=0}^{\omega-1}|a_u^*A_T (I-A_T^*A_T)^k\operatorname{sgn}(x)_T| > 1/4\bigg\}\\
\nonumber
&\leq\sum_{u\in T^c}\sum_{k=0}^{\omega-1}\mathbb{P}\Big\{|a_u^*A_T (I-A_T^*A_T)^k\operatorname{sgn}(x)_T| \geq \beta_k\Big\}\\
\label{eq.bound on prob of E3c}
&\leq\sum_{u\in T^c}\sum_{k=0}^{\omega-1}\beta_k^{-p_k}\cdot\mathbb{E}|a_u^*A_T (I-A_T^*A_T)^k\operatorname{sgn}(x)_T|^{p_k},
\end{align}
where the last step follows from Markov's inequality.
For simplicity, we select $\beta_k:=5^{-(k+1)}$.
Put $\eta_k:=(k+1)Lp_k$.
Then Lemma~\ref{lem.moment hell} gives that the $k$-th term of \eqref{eq.bound on prob of E3c} satisfies
\begin{align*}
\beta_k^{-p_k}\cdot\mathbb{E}|a_u^*A_T (I-A_T^*A_T)^k\operatorname{sgn}(x)_T|^{p_k}
&\leq 5^{(k+1)p_k} \cdot 2\eta_k^{2\eta_k}(s^{1/L}n^{-1})^{\eta_k/2}\\
&= 2\operatorname{exp}\bigg( 2\eta_k\cdot 
\log\bigg(5^{1/(2L)}\cdot \eta_k \cdot (s^{1/L}n^{-1})^{1/4}\bigg)\bigg),
\end{align*}
and the right-hand side is minimized when the inner logarithm equals $-1$, i.e., when
\begin{equation}
\label{eq.eta choice}
\eta_k
=e^{-1}5^{-1/(2L)}n^{1/4}s^{-1/(4L)}.
\end{equation}
In particular, the optimal choice of $\eta_k$ does not depend on $k$.
It remains to verify that $p_k\geq2$ for every $k$ (so that Lemma~\ref{lem.moment hell} applies) and that $\eqref{eq.bound on prob of E3c}\leq \epsilon/3$.
Since $\epsilon<1$, we have
\[
n
\geq C_Ls^{1/L}\log^4(N/\epsilon)
\geq C_Ls^{1/L}\log^4 N,
\]
which combined with \eqref{eq.eta choice} and $k+1\leq\omega\leq 3\log N$ implies that $p_k\geq 2$.
%
Also,
\begin{align*}
\eqref{eq.bound on prob of E3c}
&\leq N\cdot \omega \cdot 2e^{-2\eta_0}
= \operatorname{exp}\bigg(\log N + \log\omega + \log 2-2e^{-1}5^{-1/(2L)}n^{1/4}s^{-1/(4L)}\bigg),
\end{align*}
and the right-hand side is less than $\epsilon/3$ since $n\geq C_Ls^{1/L}\log^4(N/\epsilon)$.
\end{proof}

\section{Proof of Lemma~\ref{conditioningtheorem}}

Recall that the rows in the submatrix $A_T$ are not statistically independent.
To overcome this deficiency, our proof of Lemma~\ref{conditioningtheorem} makes use of the following decoupling estimate:

\begin{proposition}[Theorem~$3.4.1$ in~\cite{{DelapenaG:12}}]
\label{decouplingtheorem}
Fix $k\in\mathbb{N}$.
There exists a constant $C_{k}>0$ depending only on $k$ such that the following holds.
Given independent random variables $X_{1},\ldots,X_{n}$ in a measurable space $\mathcal{S}$, a separable Banach space $B$, and a measurable function $h\colon\mathcal{S}^{k}\rightarrow B$, then for every $t > 0$, it holds that
\[
\mathbb{P}\bigg\{
\bigg\|\sum_{(i_{1},\ldots,i_{k})\in\mathcal{I}_{n}^{k}} h(X_{i_{1}},\ldots, X_{i_{k}})\bigg\|_B > t\bigg\}
\leq C_{k}\cdot\mathbb{P}\bigg\{C_{k}\bigg\|\sum_{(i_{1},\ldots,i_{k})\in\mathcal{I}_{n}^{k}} h(X_{i_{1}}^{(1)},\ldots, X_{i_{k}}^{(k)})\bigg\|_B > t\bigg\},
\]
where for each $i\in [n]$ and $\ell\in [k]$, $X^{(\ell)}_{i}$ is an independent copy of $X_{i}$.  
\end{proposition}

Proposition~\ref{decouplingtheorem} will allow us to reduce Lemma~\ref{conditioningtheorem} to the following simpler result, which uses some notation that will be convenient throughout this section:
For a fixed $T\subseteq[N]$, let $X_{0}$ denote the matrix $\mathbf{11}^* - I \in \mathbb{C}^{T\times T}$, and for each $x\in[N]$, let $D_{x}\in\mathbb{C}^{T\times T}$ denote the diagonal matrix whose $t$-th diagonal entry is given by $e_N(-xt)$.

\begin{lemma}
\label{BernsteinLemma2}
Fix $L\in\mathbb{N}$.
There exists a constant $\tilde{C}^{(2)}_{L}$ depending only on $L$ such that the following holds.
Select any $s,n,N\in\mathbb{N}$ such that $N> n^L\geq s\geq 1$, and any $T\subseteq[N]$ with $|T|=s$.
Draw $\{b_i^{(j)}\}_{i\in[n],j\in[L]}$ independently with uniform distribution over $[N]$.
Then for each $\ell\in[L]$ and $\alpha\in (0,L^{-1})$, it holds that
\[
\mathbb{P}\bigg\{\bigg\| \sum_{i_{1},\ldots,i_{j}\in [n]}D_{b^{(\ell)}_{i_\ell}}\cdots D_{b^{(1)}_{i_1}}X_{0}D_{b^{(1)}_{i_1}}^{*}\cdots D_{b^{(\ell)}_{i_\ell}}^{*}\bigg\|_{2\to 2} > \tilde{C}^{(2)}_{L}n^{L/2}s^{1/2}\log^{3L/2}(N/\alpha)\bigg\}
\leq L\alpha.
\]
\end{lemma}

We prove Lemma~\ref{BernsteinLemma2} by an iterative application of both the complex Hoeffding and matrix Bernstein inequalities:

\begin{proposition}[matrix Bernstein, Theorem~1.4 in \cite{Tropp:12}]
\label{matrixbernstein}
Suppose $X_{1},\ldots, X_{n}$ are independent random $d\times d$ self-adjoint matrices with
\[
\mathbb{E}X_i=0
\qquad
\text{and}
\qquad
\|X_i\|_{2\to2}\leq K
\qquad
\text{almost surely}.
\]
Then for every $t\geq0$, it holds that
\[
\mathbb{P}\bigg\{\bigg\|\sum_{i=1}^{n}X_{i}\bigg\|_{2\to 2} > t\bigg\}
\leq 2d\operatorname{exp}\bigg(-\frac{t^{2}}{2\sigma^{2} + \frac{2}{3}Kt}\bigg),
\qquad
\sigma^{2} := \bigg\|\sum_{i=1}^{n}\mathbb{E}X_{i}^{2}\bigg\|_{2\to 2}.
\]
\end{proposition}

\begin{proof}[Proof of Lemma~\ref{BernsteinLemma2}]
Iteratively define $X_{\ell} := \sum_{j=1}^{n}D_{b_{j}^{(\ell)}}X_{\ell-1}D_{b_{j}^{(\ell)}}^{*}$.
Then our task is to prove
\begin{equation}
\label{eq.desired result}
\mathbb{P}\{\|X_{\ell}\|_{2\to 2} > \tilde{C}^{(2)}_{L} n^{L/2}s^{1/2}\log^{3L/2}(N/\alpha)\}
\leq L\alpha
\end{equation}
for every $\ell\in [L]$ and $\alpha\in(0,L^{-1})$.
To accomplish this, first put
\begin{align*}
\sigma^{2}_{\ell}(X_\ell)
&:=\bigg\|\sum_{j=1}^{n}\mathbb{E}\Big[~(D_{b_{j}^{(\ell+1)}}X_{\ell}D_{b_{j}^{(\ell+1)}}^{*})^{2}~\Big|~X_{\ell}~\Big]\bigg\|_{2\to2}\\
&=\bigg\|\sum_{j=1}^{n}\mathbb{E}\Big[~D_{b_{j}^{(\ell+1)}}X_{\ell}^2D_{b_{j}^{(\ell+1)}}^{*}~\Big|~X_{\ell}~\Big]\bigg\|_{2\to2}
=\|n\cdot\operatorname{diag}(\operatorname{diag}(X_\ell^2))\|_{2\to2}
=n\max_{t\in T} \|X_\ell e_t\|_2^2,
\end{align*}
where the last step uses the fact that $X_\ell$ is self-adjoint.
Next, fix thresholds $u_{0},\ldots,u_{L} > 0$ and $v_{0},\ldots,v_{L-1} > 0$ (to be determined later), and define events
\[
A_{\ell}
:=\{\|X_{\ell}\|_{2\to 2} \leq u_{\ell}\},
\qquad
B_{\ell} 
:= \{\sigma^2_{\ell}(X_\ell) \leq v_\ell\},
\qquad
E_\ell:=\bigcap_{i=0}^\ell B_i,
\qquad
E:=E_{L-1}.
\]
In this notation, our task is to bound $\mathbb{P}(A_{\ell}^c)$ for each $\ell\in[L]$.

First, we take $v_0=ns$ so that $E_0^c=B_0^c$ is empty, i.e., $\mathbb{P}(E_{0}^c)=0$.
Now fix $\ell\geq0$ and suppose $\mathbb{P}(E_{\ell}^c)\leq \ell\alpha/2$.
Then we can condition on $E_\ell$, and
\begin{align}
\mathbb{P}(E_{\ell+1}^c)
\nonumber
&=\mathbb{P}(E_{\ell+1}^c|E_\ell)\mathbb{P}(E_\ell)+\mathbb{P}(E_{\ell+1}^c \cap E_\ell^c)\\
\label{eq.event stuff 2}
&\leq \mathbb{P}(E_{\ell+1}^c|E_\ell)+\mathbb{P}(E_\ell^c)
=\mathbb{P}(B_{\ell+1}^c|E_\ell)+\mathbb{P}(E_\ell^c).
\end{align}
Later, we will apply Proposition~\ref{prop.complex Hoeffding} to obtain the bound
\begin{equation}
\label{eq.event stuff 3}
\mathbb{P}(B_{\ell+1}^c|E_\ell)
\leq \alpha/2,
\end{equation}
which combined with \eqref{eq.event stuff 2} implies that $\mathbb{P}(E_{\ell+1}^c)
\leq (\ell+1)\alpha/2$.
By induction, we have
\begin{equation}
\label{eq.event stuff 100}
\mathbb{P}(E^c)
\leq L\alpha/2,
\end{equation}
and so we can condition on $E$.
Next, we take $u_0=s$ so that $A_0^c$ is empty, i.e., $\mathbb{P}(A_{0}^c | E)=0$.
Now fix $\ell\geq0$ and suppose $\mathbb{P}(A_{\ell}^c | E)\leq \ell\alpha/2$.
Then we can condition on $A_{\ell}\cap E$, and
\begin{align}
\mathbb{P}(A_{\ell+1}^c | E)
\nonumber
&=\mathbb{P}(A_{\ell+1}^c | A_{\ell}\cap E)\mathbb{P}(A_{\ell} | E) + \mathbb{P}(A_{\ell+1}^c \cap A_{\ell}^c | E)\\
\label{eq.event stuff 10}
&\leq \mathbb{P}(A_{\ell+1}^c | A_{\ell}\cap E) + \mathbb{P}(A_{\ell}^c | E).
\end{align}
Later, we will apply Proposition~\ref{matrixbernstein} to obtain the bound
\begin{equation}
\label{eq.event stuff 5}
\mathbb{P}(A_{\ell+1}^c | A_{\ell}\cap E)
\leq \alpha/2,
\end{equation}
which combined with \eqref{eq.event stuff 10} implies
\begin{equation}
\label{eq.event stuff 101}
\mathbb{P}(A_{\ell+1}^c | E)
\leq (\ell+1)\alpha/2.
\end{equation}
Combining \eqref{eq.event stuff 100} and \eqref{eq.event stuff 101} then gives
\[
\mathbb{P}(A_{\ell+1}^c) 
\nonumber
= \mathbb{P}(A_{\ell+1}^c | E)\mathbb{P}(E)+\mathbb{P}(A_{\ell+1}^c \cap E^c)
\leq \mathbb{P}(A_{\ell+1}^c | E)+\mathbb{P}(E^c)
\leq L\alpha,
\]
as desired.
Overall, to prove \eqref{eq.desired result}, it suffices to select thresholds
\begin{equation}
\label{eq.sequence of u}
0
<u_1
\leq u_2
\leq \cdots
\leq u_L
\leq \tilde{C}^{(2)}_{L} n^{L/2}s^{1/2}\log^{3L/2}(N/\alpha)
\end{equation}
and $v_1,\ldots,v_{L-1}>0$ in such a way that \eqref{eq.event stuff 3} and \eqref{eq.event stuff 5} hold simultaneously.

In what follows, we demonstrate \eqref{eq.event stuff 3}.
Let $\mathbb{P}_\ell$ denote the probability measure obtained by conditioning on the event $E_\ell$, and let $\mathbb{E}_\ell$ denote expectation with respect to this measure.
As we will see, Proposition~\ref{prop.complex Hoeffding} implies that $\mathbb{P}_\ell(B_{\ell+1}^c | X_\ell)\leq \alpha/2$ holds $\mathbb{P}_\ell$-almost surely, which in turn implies
\[
\mathbb{P}(B_{\ell+1}^c | E_{\ell})
=\mathbb{E}_\ell[ \mathbb{P}_\ell(B_{\ell+1}^c | X_\ell)]
\leq \alpha/2
\]
by the law of total probability.
For each $t\in T$ and $x\in[N]$, let $D'_{t,x}\in\mathbb{C}^{T\times T}$ denote the diagonal matrix whose $t$-th diagonal entry is $0$, and whose $u$-th diagonal entry is $e_N(-xu)$ for $u\in T\setminus\{t\}$.
In particular, $D'_{t,x}$ is identical to $D_x$, save the $t$-th diagonal entry.
Then since the $t$-th entry of $X_{\ell} e_t$ is $0$, we may write
\[
X_{\ell+1}e_t
= \sum_{j=1}^{n}D_{b_{j}^{(\ell+1)}}X_{\ell}D_{b_{j}^{(\ell+1)}}^{*}e_t
= \sum_{j=1}^{n} e_N(b_{j}^{(\ell+1)}t) \cdot D_{b_{j}^{(\ell+1)}}X_{\ell} e_t
= \sum_{j=1}^{n} e_N(b_{j}^{(\ell+1)}t) \cdot D'_{t,b_{j}^{(\ell+1)}}X_{\ell} e_t.
\]
We then take norms to get
\[
\|X_{\ell+1}e_t\|_2^2
\leq \bigg\|\sum_{j=1}^{n} e_N(b_{j}^{(\ell+1)}t) \cdot D'_{t,b_{j}^{(\ell+1)}}\bigg\|_{2\to2}^2\|X_{\ell} e_t\|_2^2
= \max_{u\in T\setminus\{t\}} \bigg|\sum_{j=1}^{n} e_N(b_{j}^{(\ell+1)}(t-u)) \bigg|^2\|X_{\ell} e_t\|_2^2.
\]
Recalling the definition of $\sigma^2_{\ell}(X_\ell)$, it follows that
\begin{align*}
\sigma^2_{\ell+1}(X_{\ell+1})
&=n\cdot\max_{t\in T}\|X_{\ell+1}e_t\|_2^2\\
&\leq \sigma^2_{\ell}(X_\ell)\max_{\substack{t,u\in T\\t\neq u}} \bigg|\sum_{j=1}^{n} e_N(b_{j}^{(\ell+1)}(t-u)) \bigg|^2
\leq v_\ell \max_{\substack{t,u\in T\\t\neq u}} \bigg|\sum_{j=1}^{n} e_N(b_{j}^{(\ell+1)}(t-u)) \bigg|^2,
\end{align*}
where the last step holds $\mathbb{P}_\ell$-almost surely since $B_\ell\supseteq E_\ell$.
As such, the following bound holds $\mathbb{P}_\ell$-almost surely:
\begin{align*}
\mathbb{P}_\ell(B_{\ell+1}^c | X_\ell)
=\mathbb{P}_\ell(\{\sigma_{\ell+1}^2(X_{\ell+1}) > v_{\ell+1}\} | X_\ell)
&\leq\mathbb{P}_\ell\bigg\{\max_{\substack{t,u\in T\\t\neq u}} \bigg|\sum_{j=1}^{n} e_N(b_{j}^{(\ell+1)}(t-u)) \bigg|^2 > \frac{v_{\ell+1}}{v_\ell} \bigg\}\\
&\leq \sum_{\substack{t,u\in T\\t\neq u}}\mathbb{P}_\ell\bigg\{ \bigg|\sum_{j=1}^{n} e_N(b_{j}^{(\ell+1)}(t-u)) \bigg|^2 > \frac{v_{\ell+1}}{v_\ell} \bigg\},
\end{align*}
where the last step applies the union bound.
Next, for each $t,u\in T$ with $t\neq u$, it holds that $\{e_N(b_{j}^{(\ell+1)}(t-u))\}_{j\in[n]}$ are $\mathbb{P}_\ell$-independent complex random variables with zero mean and unit modulus $\mathbb{P}_\ell$-almost surely.
We may therefore apply Proposition~\ref{prop.complex Hoeffding} to continue:
\[
\mathbb{P}_\ell(B_{\ell+1}^c | X_\ell)
\leq 4s^2 \operatorname{exp}\bigg(-\frac{v_{\ell+1}}{2nv_\ell}\bigg)
\qquad
\text{$\mathbb{P}_\ell$-almost surely}.
\]
As such, selecting $v_{\ell+1}:=2n\log(8s^2/\alpha)\cdot v_\ell$ ensures that $\mathbb{P}_\ell(B_{\ell+1}^c | X_\ell)\leq \alpha/2$ holds $\mathbb{P}_\ell$-almost surely, as desired.

Next, we demonstrate \eqref{eq.event stuff 5}.
For convenience, we change the meanings of $\mathbb{P}_\ell$ and $\mathbb{E}_\ell$:
Let $\mathbb{P}_\ell$ denote the probability measure obtained by conditioning on the event $A_\ell\cap E$, and let $\mathbb{E}_\ell$ denote expectation with respect to this measure.
As we will see, Proposition~\ref{matrixbernstein} implies that $\mathbb{P}_\ell(A_{\ell+1}^c | X_\ell)\leq \alpha/2$ holds $\mathbb{P}_\ell$-almost surely, which in turn implies
\[
\mathbb{P}(A_{\ell+1}^c | A_{\ell}\cap E)
=\mathbb{E}_\ell[ \mathbb{P}_\ell(A_{\ell+1}^c | X_\ell)]
\leq \alpha/2
\]
by the law of total probability.
Let $\mathcal{X}\subseteq\mathbb{C}^{T\times T}$ denote the support of the discrete distribution of $X_\ell$ under $\mathbb{P}_\ell$.
For each $x\in\mathcal{X}$, let $\mathbb{P}_{\ell,x}$ denote the probability measure obtained by conditioning $\mathbb{P}_\ell$ on the event $\{X_\ell=x\}$.
Then $\{D_{b_{j}^{(\ell+1)}}X_\ell D_{b_{j}^{(\ell+1)}}^{*}\}_{j\in[n]}$ are $\mathbb{P}_{\ell,x}$-independent random $s\times s$ self-adjoint matrices with
\[
\mathbb{E}_{\ell,x}D_{b_{j}^{(\ell+1)}}X_{\ell}D_{b_{j}^{(\ell+1)}}^{*}=0,
\qquad
\|D_{b_{j}^{(\ell+1)}}X_{\ell}D_{b_{j}^{(\ell+1)}}^{*}\|_{2\to2}\leq \|x\|_{2\to2}
\qquad
\text{$\mathbb{P}_{\ell,x}$-almost surely},
\]
and so Proposition~\ref{matrixbernstein} gives
\begin{align*}
\mathbb{P}_\ell(A_{\ell+1}^c | \{X_\ell=x\} )
&=\mathbb{P}_\ell(\{\|X_{\ell+1}\|_{2\to2}> u_{\ell+1}\} | \{X_\ell=x\} )\\
&=\mathbb{P}_{\ell,x}\bigg\{\bigg\|\sum_{j=1}^n D_{b_{j}^{(\ell+1)}}X_\ell D_{b_{j}^{(\ell+1)}}^{*} \bigg\|_{2\to2} > u_{\ell+1} \bigg\}\\
&\leq 2s\operatorname{exp}\bigg(-\frac{u_{\ell+1}^2}{2\sigma_\ell^2(x)+\frac{2}{3}\|x\|_{2\to2}u_{\ell+1}}\bigg)
\leq 2s\operatorname{exp}\bigg(-\frac{u_{\ell+1}^2}{2v_\ell+\frac{2}{3}u_\ell u_{\ell+1}}\bigg),
\end{align*}
where the last step follows from the fact that $x$ resides in the support of the $\mathbb{P}_\ell$-distribution of $X_\ell$, and furthermore $\sigma^2_\ell(X_\ell)\leq v_\ell$ holds over $B_\ell \supseteq A_\ell\cap E$ and $\|X_\ell\|_{2\to2} \leq u_\ell$ holds over $A_\ell \supseteq A_\ell\cap E$.
Thus, 
\[
\mathbb{P}_\ell(A_{\ell+1}^c | X_\ell )
\leq 2s\operatorname{exp}\bigg(-\frac{u_{\ell+1}^2}{2v_\ell+\frac{2}{3}u_\ell u_{\ell+1}}\bigg)
\leq 2s\operatorname{exp}\bigg(-\frac{1}{4}\min\bigg\{\frac{u_{\ell+1}^2}{v_{\ell}},\frac{3u_{\ell+1}}{u_\ell}\bigg\}\bigg)
\]
holds $\mathbb{P}_\ell$-almost surely, and so we select $u_{\ell+1}$ so that the right-hand side is at most $\alpha/2$.
One may show that it suffices to take $u_1:=v_{L-1}^{1/2}\cdot 4\log(4s/\alpha)$ and $u_{\ell+1}=u_\ell\cdot(4/3)\log(4s/\alpha)$ for $\ell\geq 1$.
This choice satisfies \eqref{eq.sequence of u}, from which the result follows.
\end{proof}

\begin{proof}[Proof of Lemma~\ref{conditioningtheorem}]
For indices $i_{1},\ldots,i_{\ell}\in [n]$ we define
\[
H(b_{i_{1}},\ldots,b_{i_{\ell}})
:=H_{i_{1}\cdots i_{\ell}}
:= D_{b_{i_\ell}}\cdots D_{b_{i_1}}X_{0}D_{b_{i_1}}^*\cdots D_{b_{i_\ell}}^*.
\]
Notice that $H(\cdot)$ is a deterministic function, but $H_{i_{1}\cdots i_{\ell}}$ is random since $b_{i_{1}},\ldots,b_{i_{\ell}}$ are random.
Observe that we may decompose $A_{T}^{*}A_{T} - I$ as 
\[
A_{T}^{*}A_{T} - I
= \frac{1}{n^L}\sum_{i_{1},\ldots, i_{L}\in [n]} H_{i_{1}\cdots i_{L}}.
\]
For each $\ell\in[L]$, it will be convenient to partition $[n]^\ell=\mathcal{I}_{n}^{\ell}\sqcup\mathcal{J}_{n}^{\ell}$, where $\mathcal{I}_{n}^{\ell}$ denotes the elements with distinct entries and $\mathcal{J}_{n}^{\ell}$ denotes the elements that do not have distinct entries.
In addition, we let $S([L],\ell)$ denote the set of partitions of $[L]$ into $\ell$ nonempty sets, and we put $S(L,\ell):=|S([L],\ell)|$.
To analyze the above sum, we will relate $A_{T}^{*}A_{T} - I$ to a sum indexed over $\mathcal{I}_{n}^{\ell}$.  
Specifically, every tuple $(i_{1},\ldots, i_{L})\in [n]^{L}$ is uniquely determined by three features:
\begin{itemize}
\item
the number $\ell\in [L]$ of distinct entries in the tuple,
\item
a partition $S\in S([L],\ell)$ of the tuple indices into $\ell$ nonempty sets, and
\item
a tuple $(j_{1},\ldots, j_{\ell})\in \mathcal{I}_{n}^{\ell}$ of distinct elements of $[n]$.   
\end{itemize}
To be clear, we order the members of $S=\{S_1,\ldots,S_\ell\}$ lexicographically so that $i_k=j_l$ for every $k\in S_l$.
For every $S\in S([L],\ell)$, we may therefore abuse notation by considering the function $S\colon[n]^\ell\to[n]^L$ defined by $S(j_{1},\ldots, j_{\ell}) = (i_{1},\ldots,i_{L})$.
It will also be helpful to denote $\sigma\colon[L]\to[\ell]$ such that for every $i\in[L]$, it holds that $i\in S_{\sigma(i)}$, that is, $i$ resides in the $\sigma(i)$-th member of the partition $S$.
This gives
\[
A_{T}^{*}A_{T} - I
=\frac{1}{n^L} \sum_{i_{1},\ldots, i_{L}\in [n]} H_{i_{1}\cdots i_{L}}
=\frac{1}{n^L} \sum_{\ell=1}^{L}\sum_{S\in S([L],\ell)} \sum_{(j_{1},\ldots,j_{\ell})\in \mathcal{I}_{n}^{\ell}} H_{S(j_{1},\ldots,j_{\ell})}.
\]    
The triangle inequality and union bound then give
\begin{align*}
E
:=\mathbb{P}\bigg\{\|A_{T}^{*}A_{T} - I\|_{2\rightarrow 2} > \frac{1}{e}\bigg\}
&= \mathbb{P}\bigg\{\bigg\|\sum_{\ell=1}^{L}\sum_{S\in S([L],\ell)} \sum_{(j_{1},\ldots,j_{\ell})\in \mathcal{I}_{n}^{\ell}} H_{S(j_{1},\ldots,j_{\ell})}\bigg\|_{2\to2} > \frac{n^{L}}{e}\bigg\}\\
&\leq \mathbb{P}\bigg\{\sum_{\ell=1}^{L}\sum_{S\in S([L],\ell)} \bigg\|\sum_{(j_{1},\ldots,j_{\ell})\in \mathcal{I}_{n}^{\ell}} H_{S(j_{1},\ldots,j_{\ell})}\bigg\|_{2\to2} > \frac{n^{L}}{e}\bigg\}\\
&\leq \sum_{\ell=1}^{L}\sum_{S\in S([L],\ell)} \mathbb{P}\bigg\{\bigg\|\sum_{(j_{1},\ldots,j_{\ell})\in \mathcal{I}_{n}^{\ell}} H_{S(j_{1},\ldots, j_{\ell})}\bigg\|_{2\to2} > \frac{n^{L}}{eL2^{L^2}}\bigg\},
\end{align*}
where the last inequality uses the simple bound $S(L,\ell)\leq 2^{L^2}$. 

To proceed, we will apply Proposition~\ref{decouplingtheorem} with $k=\ell$, random variables $X_i=b_{i}$ for $i\in[n]$, measurable space $\mathcal{S}=[N]$, separable Banach space $B=(\mathbb{C}^{T\times T},\|\cdot\|_{2\to2})$, and measurable function $h_S\colon[N]^\ell\to\mathbb{C}^{T\times T}$ defined by
\[
h_S(x_1,\ldots,x_\ell)
:=H(x_{\sigma(1)},\ldots,x_{\sigma(L)})
=H(s_1x_1,\ldots,s_\ell x_\ell),
\]
where $s_k:=|S_k|$ for $k\in[\ell]$.
Observe that $H_{S(j_1,\ldots, j_\ell)}=h_S(b_{j_1},\ldots,b_{j_\ell})$.
As such, we may continue our bound:
\begin{align*}
E
&\leq \sum_{\ell=1}^{L}\sum_{S\in S([L],\ell)} \mathbb{P}\bigg\{\bigg\|\sum_{(j_{1},\ldots,j_{\ell})\in \mathcal{I}_{n}^{\ell}} H_{S(j_{1},\ldots, j_{\ell})}\bigg\|_{2\to2} > \frac{n^{L}}{eL2^{L^2}}\bigg\}\\
&= \sum_{\ell=1}^{L}\sum_{S\in S([L],\ell)} \mathbb{P}\bigg\{\bigg\|\sum_{(j_{1},\ldots,j_{\ell})\in \mathcal{I}_{n}^{\ell}} h_S(b_{j_1},\ldots,b_{j_\ell}) \bigg\|_{2\to2} > \frac{n^{L}}{eL2^{L^2}}\bigg\}\\
&\leq \sum_{\ell=1}^{L}\sum_{S\in S([L],\ell)} C_\ell \cdot \mathbb{P}\bigg\{C_\ell\bigg\|\sum_{(j_{1},\ldots,j_{\ell})\in \mathcal{I}_{n}^{\ell}} h_S(b_{j_1}^{(1)},\ldots,b_{j_\ell}^{(\ell)}) \bigg\|_{2\to2} > \frac{n^{L}}{eL2^{L^2}}\bigg\}\\
&= \sum_{\ell=1}^{L}\sum_{S\in S([L],\ell)} C_\ell \cdot \mathbb{P}\bigg\{C_\ell\bigg\|\sum_{(j_{1},\ldots,j_{\ell})\in \mathcal{I}_{n}^{\ell}} H(s_1b_{j_1}^{(1)},\ldots,s_\ell b_{j_\ell}^{(\ell)}) \bigg\|_{2\to2} > \frac{n^{L}}{eL2^{L^2}}\bigg\}.
\end{align*}
Since $L<N$, then for each $S\in S([L],\ell)$, we have $s_i\in(0,N)$ for every $i\in[\ell]$.
Also, $N$ is prime, and so the mapping $M\colon\mathbb{F}_N^{\ell\times n}\to\mathbb{F}_N^{\ell \times n}$ defined by $M(X)=\operatorname{diag}(s_1,\ldots,s_\ell)\cdot X$ is invertible.
Thus, $(s_ib_j^{(i)})_{i\in[\ell],j\in[n]}$ has the same (uniform) distribution as $(b_j^{(i)})_{i\in[\ell],j\in[n]}$.
Put
\[
H'_{j_{1}\cdots j_{\ell}}
:= D_{b^{(\ell)}_{j_\ell}}\cdots D_{b^{(1)}_{j_1}}X_{0}D_{b^{(1)}_{j_1}}^{*}\cdots D_{b^{(\ell)}_{j_\ell}}^{*}.
\]
Denoting $C:=\max_{\ell\leq L}C_\ell$ in terms of the absolute constants of Proposition~\ref{decouplingtheorem}, then
\[
E
\leq 2^{L^2}C\cdot\sum_{\ell=1}^{L}\mathbb{P}\bigg\{\bigg\|\sum_{(j_{1},\ldots,j_{\ell})\in \mathcal{I}^{\ell}_{n}} H'_{j_{1}\cdots j_{\ell}}\bigg\|_{2\to2} > \frac{n^{L}}{CeL2^{L^2}}\bigg\}.
\]

We will use Lemma~\ref{BernsteinLemma2} to bound each of the probabilities in the above sum.
However, the sums in Lemma~\ref{BernsteinLemma2} are indexed by $[n]^{\ell}$ instead of $\mathcal{I}_{n}^{\ell}$.  
To close this gap, we perform another sequence of union bounds.  
First, 
\begin{align}
E
\nonumber
&\leq 2^{L^2}C\cdot \sum_{\ell=1}^{L}\mathbb{P}\bigg\{ \bigg\| \sum_{j_{1},\ldots,j_{\ell}\in [n]} H'_{j_{1}\cdots j_{\ell}} - \sum_{(j_{1},\ldots, j_{\ell})\in \mathcal{J}^{\ell}_{n}} H_{j_{1}\cdots j_{\ell}}' \bigg\|_{2\to2} > \frac{n^{L}}{CeL2^{L^2}}\bigg\}\\
\label{eq.first part of this}
&\leq 2^{L^2}C\cdot \sum_{\ell=1}^{L}\mathbb{P}\bigg\{ \bigg\| \sum_{j_{1},\ldots,j_{\ell}\in [n]} H'_{j_{1}\cdots j_{\ell}} \bigg\|_{2\to2} > \frac{n^{L}}{2CeL2^{L^2}} \bigg\} \\
\label{eq.next part of this}
&\qquad + 2^{L^2}C\cdot \sum_{\ell=1}^{L}\mathbb{P}\bigg\{ \bigg\|\sum_{(j_{1},\ldots, j_{\ell})\in \mathcal{J}^{\ell}_{n}} H'_{j_{1}\cdots j_{\ell}}\bigg\|_{2\to2}  > \frac{n^{L}}{2CeL2^{L^2}}\bigg\}.
\end{align}
Next, we will decompose the sum over $\mathcal{J}_{n}^{\ell}$ into multiples of sums of over $[n]^{1},\ldots, [n]^{\ell-1}$.  
Each $(j_{1},\ldots,j_{\ell})\in \mathcal{J}_{n}^{\ell}$ induces a partition of $[\ell]$ into at most $\ell-1$ parts.  
We may rewrite this sum by first summing over each partition $S$ into $\ell-1$ parts and then over each assignment of elements to this partition, represented by a member of $[n]^{\ell-1}$.  
That is, we may consider
\[
\sum_{S\in S([\ell],\ell-1)} \sum_{i_{1},\ldots,i_{\ell-1}\in [n]} H'_{S(i_{1},\ldots,i_{\ell-1})}.
\]
The above expression will include $H'_{j_{1}\cdots j_{\ell}}$ exactly once for each $(j_{1},\ldots,j_{\ell})$ that induces a partition of $[\ell]$ into exactly $\ell-1$ parts.  
However, since the indices $(i_{1},\ldots,i_{\ell-1})\in [n]^{\ell-1}$ need not be distinct, for each partition into more than $\ell-1$ parts, the above sum will overcount the corresponding terms $H'_{j_{1}\cdots j_{\ell}}$.  
To remedy this, we subtract the appropriate multiple of the sum over all partitions into $\ell-2$ parts.  
Then, having subtracted too many terms of the form $H_{j_{1}\cdots j_{\ell}}$ corresponding to partitions into $\ell-3$ parts, we may add back the appropriate number of such terms.  
Iterating this process leads to a decomposition of the form
\[
\sum_{(j_{1},\ldots, j_{\ell})\in \mathcal{J}^{\ell}_{n}} H'_{j_{1}\cdots j_{\ell}}
= \sum_{j=1}^{\ell-1}(-1)^{\ell-1-j}\sum_{\substack{1\leq s_{1},\ldots,s_{j} \leq \ell \\ s_{1} + \cdots + s_{j} = \ell}} c_{s_{1},\ldots,s_{j}}\sum_{\substack{S\in S([\ell],j)\\ |S_1|=s_1,\ldots, |S_j|=s_j}} \sum_{i_{1},\ldots,i_{j}\in [n]} H'_{S(i_{1},\ldots,i_{j})}.
\]
For our purposes, it suffices to observe that $c_{s_{1},\ldots,s_{j}}$ is some positive integer crudely bounded by $2^{L^3}$; indeed, $2^{L^2}$ is a bound on the number of partitions of $[\ell]$, and there are at most $L-1$ steps in the iterative process above, so
\[
c_{s_{1},\ldots,s_{j}}
\leq 1+2^{L^2}+2^{2L^2}+\cdots+2^{(L-1)L^2}
\leq 2^{L^3}.
\]
The triangle inequality then gives
\begin{align*}
\bigg\|\sum_{(j_{1},\ldots, j_{\ell})\in \mathcal{J}^{\ell}_{n}} H'_{j_{1}\cdots j_{\ell}}\bigg\|_{2\to2}
&\leq \sum_{j=1}^{\ell-1}\sum_{\substack{1\leq s_{1},\ldots,s_{j} \leq \ell \\ s_{1} + \cdots + s_{j} = \ell}} \sum_{\substack{S\in S([\ell],j)\\ |S_1|=s_1,\ldots, |S_j|=s_j}}  c_{s_{1},\ldots,s_{j}} \bigg\| \sum_{i_{1},\ldots,i_{j}\in [n]} H'_{S(i_{1},\ldots,i_{j})}\bigg\|_{2\to2}\\
&\leq 2^{L^3} \cdot \sum_{j=1}^{\ell-1} \sum_{S\in S([\ell],j)}\bigg\| \sum_{i_{1},\ldots,i_{j}\in [n]} H'_{S(i_{1},\ldots,i_{j})}\bigg\|_{2\to2}.
\end{align*}
Next, we apply a union bound:
\begin{align*}
\eqref{eq.next part of this}
&= 2^{L^2}C\cdot \sum_{\ell=1}^{L} \mathbb{P}\bigg\{ \bigg\|\sum_{(j_{1},\ldots, j_{\ell})\in \mathcal{J}^{\ell}_{n}} H'_{j_{1}\cdots j_{\ell}}\bigg\|_{2\to2}  > \frac{n^{L}}{2CeL2^{L^2}}\bigg\}\\
&\leq 2^{L^2}C \cdot\sum_{\ell=1}^{L} \mathbb{P}\bigg\{ \sum_{j=1}^{\ell-1}\sum_{S\in S([\ell],j)} \bigg\| \sum_{i_{1},\ldots,i_{j}\in [n]} H'_{S(i_{1},\ldots,i_{j})}\bigg\|_{2\to2}  > \frac{n^{L}}{2CeL2^{2L^3}}\bigg\}\\
&\leq 2^{L^2}C \cdot \sum_{\ell=1}^{L} \sum_{j=1}^{\ell-1}\sum_{S\in S([\ell],j)} \mathbb{P}\bigg\{  \bigg\| \sum_{i_{1},\ldots,i_{j}\in [n]} H'_{S(i_{1},\ldots,i_{j})}\bigg\|_{2\to2}  > \frac{n^{L}}{2CeL^22^{3L^3}}\bigg\}
\end{align*}
Let $F\in\mathbb{F}_N^{j\times \ell}$ denote the matrix whose $i$th row is the indicator vector of $S_i$.
Since the $S_i$'s are nonempty and partition $[\ell]$, it follows that $F$ has rank $j$.
Define the mapping $M\colon\mathbb{F}_N^{\ell\times n}\to\mathbb{F}_N^{j \times n}$ by $M(X)=FX$.
Then $M$ is surjective, and so $M((b_k^{(i)})_{i\in[\ell],k\in[n]})$ has the same (uniform) distribution as $(b_k^{(i)})_{i\in[j],k\in[n]}$.
Next, we take $g\colon \mathbb{F}_N^{j \times n}\to(\mathbb{C}^{T\times T})^{[n]^j}$ such that $g((x_k^{(i)})_{i\in[j],k\in[n]})_{i_1,\ldots,i_j}=H(x_{i_1}^{(1)},\ldots,x_{i_j}^{(j)})$.
Then $(H'_{S(i_{1},\ldots,i_{j})})_{i_1,\ldots,i_j\in[n]}=g(M((b_k^{(i)})_{i\in[\ell],k\in[n]}))$ has the same distribution as $(H'_{i_{1}\cdots i_{j}})_{i_1,\ldots,i_j\in[n]}=g((b_k^{(i)})_{i\in[j],k\in[n]})$.
With this, we continue:
\begin{align}
\eqref{eq.next part of this}
\nonumber
&\leq 2^{L^2}C \cdot \sum_{\ell=1}^{L} \sum_{j=1}^{\ell-1}\sum_{S\in S([\ell],j)} \mathbb{P}\bigg\{  \bigg\| \sum_{i_{1},\ldots,i_{j}\in [n]} H'_{S(i_{1},\ldots,i_{j})}\bigg\|_{2\to2}  > \frac{n^{L}}{2CeL^22^{3L^3}}\bigg\}\\
\label{eq.ready to bound}
&=2^{L^2}C \cdot \sum_{\ell=1}^{L} \sum_{j=1}^{\ell-1}\sum_{S\in S([\ell],j)} \mathbb{P}\bigg\{  \bigg\| \sum_{i_{1},\ldots,i_{j}\in [n]} H'_{i_{1}\cdots i_{j}}\bigg\|_{2\to2}  > \frac{n^{L}}{2CeL^22^{3L^3}}\bigg\}.
\end{align}
Finally, pick $\alpha:=\epsilon/(6CL^32^{2L^2})$ and take $C^{(2)}_L$ large enough so that
\[
n^{L/2}
\geq (C^{(2)}_L)^{L/2} s^{1/2}\log^{3L/2}(N/\epsilon)
\geq 2CeL^2 2^{3L^3} \cdot \tilde{C}^{(2)}_L s^{1/2}\log^{3L/2}(N/\alpha).
\]
Then by Lemma~\ref{BernsteinLemma2}, each of the terms in \eqref{eq.first part of this} and \eqref{eq.ready to bound} are at most $L\alpha$, and so
\[
E
\leq\eqref{eq.first part of this}+\eqref{eq.ready to bound}
\leq2^{L^2}C\cdot L\cdot L\alpha + 2^{L^2}C\cdot L\cdot L\cdot 2^{L^2} \cdot L\alpha
\leq 2CL^32^{2L^2}\alpha
= \epsilon/3.
\qedhere
\]
\end{proof}

\section{Proof of Lemma~\ref{lem.moment hell}}

We will use the following lemma, whose proof we save for later.

\begin{lemma}
\label{lem.linear independence}
Fix $u\in\mathbb{R}$ and $j\colon[k+1]\times[2M]\times[L]\to[n]$ with image of size $m$.
Then at least $\lceil m/L\rceil$ of the following $n$ linear constraints on $\ell\in\mathbb{R}^{(k+1)\times 2M}$ are linearly independent:
\[
\sum_{(h,p,q)\in j^{-1}(i)}(-1)^p(\ell_{h,p}-\ell_{h-1,p})=0,
\qquad
i\in[n],
\]
where $\ell_{0,p}:=u$ for every $p\in[2M]$.
\end{lemma}

\begin{proof}[Proof of Lemma~\ref{lem.moment hell}]
For now, we assume that $p$ takes the form $p=2M$ for some $M\in\mathbb{N}$, and later we interpolate with Littlewood's inequality.
As such, we will bound the expectation of
\begin{equation}
\label{eq.want the moments}
|a_u^*A_T(I-A_T^*A_T)^k z|^{2M}.
\end{equation}
In what follows, we write \eqref{eq.want the moments} as a large sum of products, and then we will take the expectation of each product in the sum.
First, we denote $\ell_0:=u$ and observe
\[
a_u^*A_T(I-A_T^*A_T)^k z
=-\sum_{\ell_1,\ldots,\ell_{k+1}\in T}\prod_{h=1}^{k+1}(I-A^*A)_{\ell_{h-1},\ell_h}z_{\ell_{k+1}}.
\]
Since $A$ has unit-norm columns, the diagonal entries of $I-A^*A$ are all zero, and so we can impose the constraint that $\ell_h\neq \ell_{h+1}$ for every $h\in[k]$.
For the moment, we let $T_*^{k+1}$ denote the set of all $(k+1)$-tuples $\ell$ of members of $T$ that satisfy this constraint.
Then
\[
a_u^*A_T(I-A_T^*A_T)^k z
=-\sum_{\ell\in T_*^{k+1}}\prod_{h=1}^{k+1}(I-A^*A)_{\ell_{h-1},\ell_h}z_{\ell_{k+1}}.
\]
Next, we take the squared modulus of this quantity and raise it to the $M$-th power.
To do so, it will be convenient to adopt the following notation to keep track of complex conjugation:
Write $\overline{x}^{(p)}$ to denote $\overline{x}$ if $p$ is odd and $x$ if $p$ is even.
Then
\begin{align*}
|a_u^*A_T(I-A_T^*A_T)^k z|^{2M}
&=\bigg|-\sum_{\ell\in T_*^{k+1}}\prod_{h=1}^{k+1}(I-A^*A)_{\ell_{h-1},\ell_h}z_{\ell_{k+1}}\bigg|^{2M}\\
&=\sum_{\ell^{(1)},\ldots,\ell^{(2M)}\in T_*^{k+1}}\prod_{p=1}^{2M}\prod_{h=1}^{k+1}\overline{(I-A^*A)_{\ell_{h-1,p},\ell_{h,p}}z_{\ell_{k+1,p}}}^{(p)},
\end{align*}
where $\ell_{0,p}=u$ and $\ell_{h,p}$ denotes the $h$-th entry of $\ell^{(p)}$ for every $p\in[2M]$.
It is helpful to think of each tuple $(\ell^{(1)},\ldots,\ell^{(2M)})\in(T_*^{k+1})^{2M}$ as a matrix $\ell\in T^{(k+1)\times2M}$ satisfying constraints of the form $\ell_{h-1,p}\neq\ell_{h,p}$.
Let $\mathcal{L}$ denote the set of these matrices.
Next, we write out the entries of $I-A^*A$ in terms of the entries of $A$:
\[
(I-A^*A)_{\ell_{h-1,p},\ell_{h,p}}
=-(A^*A)_{\ell_{h-1,p},\ell_{h,p}}
=-\frac{1}{n^L}\sum_{j_1,\ldots,j_L\in[n]} e_N((\ell_{h,p}-\ell_{h-1,p})(b_{j_1}+\cdots+b_{j_L})).
\]
After applying the distributive law, it will be useful to index with $j_{h,p,1},\ldots,j_{h,p,L}\in[n]$ in place of $j_1,\ldots,j_L\in[n]$ above, and in order to accomplish this, we let $\mathcal{J}$ denote the set of functions $j\colon[k+1]\times [2M]\times [L]\to[n]$.
We also collect the entries of $z$ by writing $z_\ell:=\prod_{p=1}^{2M}\prod_{h=1}^{k+1}\overline{z_{\ell_{k+1,p}}}^{(p)}$.
Then the distributive law gives
\begin{align*}
|a_u^*A_T(I-A_T^*A_T)^k z|^{2M}
&=\sum_{\ell\in\mathcal{L}} z_\ell \prod_{p=1}^{2M}\prod_{h=1}^{k+1}\frac{1}{n^L}\sum_{j_1,\ldots,j_L\in[n]}e_N((-1)^p(\ell_{h,p}-\ell_{h-1,p})(b_{j_{1}}+\cdots+b_{j_{L}}))\\
&=\frac{1}{n^{2ML(k+1)}}\sum_{\ell\in\mathcal{L}} z_\ell \sum_{j\in\mathcal{J}}\prod_{p=1}^{2M}\prod_{h=1}^{k+1} e_N\bigg(\sum_{q=1}^L(-1)^p(\ell_{h,p}-\ell_{h-1,p})b_{j_{h,p,q}}\bigg)\\
&=\frac{1}{n^{2ML(k+1)}}\sum_{\ell\in\mathcal{L}} z_\ell \sum_{j\in\mathcal{J}} e_N\bigg(\sum_{p=1}^{2M}\sum_{h=1}^{k+1}\sum_{q=1}^L(-1)^p(\ell_{h,p}-\ell_{h-1,p})b_{j_{h,p,q}}\bigg)\\
&=\frac{1}{n^{2ML(k+1)}}\sum_{\ell\in\mathcal{L}} z_\ell \sum_{j\in\mathcal{J}} e_N\bigg(\sum_{i\in[n]}\sum_{(h,p,q)\in j^{-1}(i)}(-1)^p(\ell_{h,p}-\ell_{h-1,p})b_{i}\bigg)\\
&=\frac{1}{n^{2ML(k+1)}}\sum_{\ell\in\mathcal{L}} z_\ell \sum_{j\in\mathcal{J}} \prod_{i\in[n]}e_N\bigg(\sum_{(h,p,q)\in j^{-1}(i)}(-1)^p(\ell_{h,p}-\ell_{h-1,p})b_{i}\bigg).
\end{align*}
When bounding the expectation, we can remove the $z_\ell$'s.
To be explicit, we apply the linearity of expectation and the independence of the $b_i$'s, the triangle inequality and the fact that $\|z\|_\infty\leq 1$, and then the fact that each expectation $\mathbb{E}e_N$ is either $0$ or $1$:
\[
\mathbb{E}\sum_\ell z_\ell \sum_j\prod_i e_N
=\sum_\ell z_\ell \sum_j\prod_i \mathbb{E} e_N
\leq \sum_\ell \bigg|\sum_j\prod_i \mathbb{E} e_N\bigg|
= \sum_\ell \sum_j\prod_i \mathbb{E} e_N.
\]
This produces the bound
\begin{align}
\nonumber
&\mathbb{E}|a_u^*A_T(I-A_T^*A_T)^k z|^{2M}\\
\nonumber
&\qquad\leq \frac{1}{n^{2ML(k+1)}}\sum_{\ell\in\mathcal{L}} \sum_{j\in\mathcal{J}} \prod_{i\in[n]}\mathbb{E}e_N\bigg(\sum_{(h,p,q)\in j^{-1}(i)}(-1)^p(\ell_{h,p}-\ell_{h-1,p})b_{i}\bigg)\\
\label{eq.bound on expectation}
&\qquad=\frac{1}{n^{2ML(k+1)}}\bigg|\bigg\{(\ell,j)\in\mathcal{L}\times\mathcal{J}:\sum_{(h,p,q)\in j^{-1}(i)}(-1)^p(\ell_{h,p}-\ell_{h-1,p})=0~~\forall i\in[n]\bigg\}\bigg|.
\end{align}
We estimate the size of this set $S$ by first bounding the number $N_1(j)$ of $\ell\in\mathcal{L}$ such that $(\ell,j)\in S$, and then bounding the sum $|S|=\sum_{j\in\mathcal{J}}N_1(j)$.
Let $m(j)$ denote the size of the image of $j\in\mathcal{J}$ and put $r:=\lceil m(j)/L\rceil$.
By Lemma~\ref{lem.linear independence}, there exist linearly independent $\{w_i\}_{i\in[r]}$ in $\mathbb{R}^{(k+1)\times 2M}$ and $\{b_i\}_{i\in[r]}$ in $\mathbb{R}$ such that $(\ell,j)\in S$ only if $\ell\in\mathbb{R}^{(k+1)\times 2M}$ satisfies $\langle \ell,w_i\rangle=b_i$ for every $i\in[r]$.
Of course, $(\ell,j)\in S$ also requires $\ell\in T^{(k+1)\times 2M}$.
As such, we use identity basis elements to complete $\{w_i\}_{i\in[r]}$ to a basis $\{w_i\}_{i\in[2M(k+1)]}$ for $\mathbb{R}^{(k+1)\times 2M}$.
Then $(\ell,j)\in S$ only if $\ell\in\mathbb{R}^{(k+1)\times 2M}$ satisfies 
\[
\langle \ell,w_i\rangle=b_i
~~~
\forall i\in\{1,\ldots,r\}
\qquad
\text{and}
\qquad
\langle \ell,w_i\rangle\in T
~~~
\forall i\in\{r+1,\ldots,2M(k+1)\}.
\]
It follows that $N_1(j)\leq |T|^{2M(k+1)-r}$.
Next, observe that $\eta$ is the size of the domain of $j\in\mathcal{J}$.
If $m(j)>\eta/2$, then there exists $i\in[n]$ for which the preimage $j^{-1}(i)$ is a singleton set $\{(h,p,q)\}$, in which case there is no $(\ell,j)\in S$ since $\ell_{h,p}\neq\ell_{h-1,p}$.
Overall,
\[
N_1(j)
\leq
\left\{
\begin{array}{cl}
s^{2M(k+1)-\lceil m(j)/L\rceil}&\text{if } m(j)\leq \eta/2\\
0&\text{else.}
\end{array}
\right.
\]
Next, let $N_2(m)$ denote the number of $j\in\mathcal{J}$ with image of size $m$.
Then $N_2(m)\leq n^m m^\eta$, since there are $\binom{n}{m}\leq n^m$ choices for the image, and for each image, there are at most $m^\eta$ choices for $j$ (we say ``at most'' here since the image of $j$ needs to have size $m$).
Then
\begin{align}
|S|
= \sum_{j\in\mathcal{J}}N_1(j)
\nonumber
&\leq \sum_{m=1}^{\eta/2} N_2(m) s^{2M(k+1)-\lceil m/L\rceil}\\
\nonumber
&\leq \sum_{m=1}^{\eta/2} n^m m^\eta (s^{1/L})^{\eta-m}\\
\label{eq.bound on S}
&\leq (\eta/2)^{\eta}(s^{1/L})^{\eta} \sum_{m=1}^{\eta/2} (ns^{-1/L})^m
\leq (\eta/2)^{\eta}(s^{1/L})^{\eta} \cdot 2(ns^{-1/L})^{\eta/2},
\end{align}
where the last step uses the fact that $x:=ns^{-1/L}\geq 2$, or equivalently $x-1\geq x/2$, which implies $\sum_{i=1}^k x^i=\frac{x^{k+1}-x}{x-1}\leq 2(x^k-1)\leq 2x^k$.
We now combine \eqref{eq.bound on expectation} and \eqref{eq.bound on S}:
\begin{equation}
\label{eq.bound for integers}
\mathbb{E}|a_u^*A_T(I-A_T^*A_T)^k z|^{2M}
\leq n^{-\eta}\cdot (\eta/2)^{\eta}(s^{1/L})^{\eta} \cdot 2(ns^{-1/L})^{\eta/2}
=2(\eta/2)^\eta(s^{1/L}n^{-1})^{\eta/2}.
\end{equation}
Finally, we interpolate using Littlewood's inequality.
Put $X:=a_u^*A_T(I-A_T^*A_T)^k z$, given any $p\geq 2$, let $M$ denote the largest integer for which $2M\leq p$, and put $\theta:=p/2-M$.
Consider the function defined by $\eta(x):=(k+1)Lx$, and put $\eta:=\eta(p)$, $\eta_1:=\eta(2M)$, and $\eta_2:=\eta(2M+2)$.
Then \eqref{eq.bound for integers} implies
\begin{align*}
\mathbb{E}|X|^p
&\leq(\mathbb{E}|X|^{2M})^{1-\theta}(\mathbb{E}|X|^{2M+2})^{\theta}\\
&\leq 2  ((\eta_1/2)^{\eta_1})^{1-\theta}((\eta_2/2)^{\eta_2})^{\theta} (s^{1/L}n^{-1})^{\eta/2}
\leq 2 \eta^{2\eta}(s^{1/L}n^{-1})^{\eta/2},
\end{align*}
where the last step applies the fact that $\eta_1,\eta_2\leq 2\eta$.
\end{proof}

\begin{proof}[Proof of Lemma~\ref{lem.linear independence}]
First, we isolate the constant terms in the left-hand side:
\[
\sum_{(h,p,q)\in j^{-1}(i)}(-1)^p(\ell_{h,p}-\ell_{h-1,p})
=\sum_{(1,p,q)\in j^{-1}(i)}(-1)^p(\ell_{1,p}-\ell_{0,p})
+\sum_{\substack{(h,p,q)\in j^{-1}(i)\\h>1}}(-1)^p(\ell_{h,p}-\ell_{h-1,p}).
\]
Since $\ell_{0,p}=u$ for every $p$, we have
\begin{equation}
\label{eq.linear constraint}
\sum_{(1,p,q)\in j^{-1}(i)}(-1)^p\ell_{1,p}+
\sum_{\substack{(h,p,q)\in j^{-1}(i)\\h>1}}(-1)^p(\ell_{h,p}-\ell_{h-1,p})
=\bigg(\sum_{(1,p,q)\in j^{-1}(i)}(-1)^{p+1}\bigg)u.
\end{equation}
Let $e_{h,p}\in\mathbb{R}^{(k+1)\times 2M}$ denote the matrix that is $1$ at entry $(h,p)$ and $0$ otherwise, and consider the basis $\{b_{h,p}\}_{h\in[k+1],p\in[2M]}$ of $\mathbb{R}^{(k+1)\times[2M]}$ defined by
\[
b_{h,p}:=\left\{\begin{array}{cl}
(-1)^p e_{1,p}&\text{if } h=1\\
(-1)^p(e_{h,p}-e_{h-1,p})&\text{else.}
\end{array}\right.
\]
Then the left-hand side of \eqref{eq.linear constraint} may be rewritten as $\langle \ell,\sum_{(h,p,q)\in j^{-1}(i)}b_{h,p}\rangle$.
It remains to find a subset $S\subseteq[n]$ of size $\lceil m/L\rceil$ for which $\{\sum_{(h,p,q)\in j^{-1}(i)}b_{h,p}\}_{i\in S}$ is linearly independent.
Initialize $S=\emptyset$, $R:=\operatorname{im}(j)$, and $t=1$, and then do the following until $R$ is empty:
\begin{itemize}
\item
select any $(h_t,p_t)$ for which there exists $q_t$ such that $i_t:=j(h_t,p_t,q_t)\in R$, and
\item
update $S\leftarrow S\cup\{i_t\}$, $R\leftarrow R\setminus\{j(h_t,p_t,q):q\in[L]\}$ and $t\leftarrow t+1$.
\end{itemize}
Since each iteration removes at most $L$ members from $R$, the resulting $S$ has size at least $\lceil |\operatorname{im}(j)|/L\rceil=\lceil m/L\rceil$.
By construction, every $t$ has the property that there is no $q\in[L]$ or $u<t$ for which $(h_t,p_t,q)\in j^{-1}(i_u)$, and it follows that $\sum_{(h,p,q)\in j^{-1}(i_t)}b_{h,p}$ is the first member of the sequence to exhibit a contribution from $b_{h_t,p_t}$.
Thanks to this triangularization, we may conclude that $\{\sum_{(h,p,q)\in j^{-1}(i)}b_{h,p}\}_{i\in S}$ is linearly independent.
\end{proof}

\section*{Acknowledgments}

The authors thank Holger Rauhut for pointing out~\cite{HugelRS:14} as a potential avenue for derandomized compressed sensing.
DGM was partially supported by AFOSR FA9550-18-1-0107 and NSF DMS 1829955.


\begin{thebibliography}{WW}

\bibitem{AmelunxenLMT:14}
D.\ Amelunxen, M.\ Lotz, M.\ B.\ McCoy, J.\ A.\ Tropp,
Living on the edge:\ Phase transitions in convex programs with random data,
Inform.\ Inference 3 (2014) 224--294.

\bibitem{ApplebaumHSC:09}
L.\ Applebaum, S.\ D.\ Howard, S.\ Searle, R.\ Calderbank,
Chirp sensing codes:\ Deterministic compressed sensing measurements for fast recovery,
Appl.\ Comput.\ Harmon.\ Anal.\ 26 (2009) 283--290.

\bibitem{BandeiraDMS:13}
A.\ S.\ Bandeira, E.\ Dobriban, D.\ G.\ Mixon, W.\ F.\ Sawin,
Certifying the restricted isometry property is hard,
IEEE Trans.\ Inform.\ Theory 59 (2013) 3448--3450.

\bibitem{BandeiraFMM:16}
A.\ S.\ Bandeira, M.\ Fickus, D.\ G.\ Mixon, J.\ Moreira,
Derandomizing restricted isometries via the Legendre symbol,
Constr.\ Approx.\ 43 (2016) 409--424.

\bibitem{BandeiraFMW:13}
A.\ S.\ Bandeira, M.\ Fickus, D.\ G.\ Mixon, P.\ Wong,
The road to deterministic matrices with the restricted isometry property,
J.\ Fourier Anal.\ Appl.\ 19 (2013) 1123--1149.

\bibitem{BandeiraMM:17}
A.\ S.\ Bandeira, D.\ G.\ Mixon, J.\ Moreira,
A conditional construction of restricted isometries,
Int.\ Math.\ Res.\ Not.\ (2017) 372--381.

\bibitem{Bourgain:14}
J.\ Bourgain,
An improved estimate in the restricted isometry problem,
In:\ Geometric Aspects of Functional Analysis, Springer, 2014, pp.\ 65--70.

\bibitem{BaraniukDDW:08}
R.\ Baraniuk, M.\ Davenport, R.\ DeVore, M.\ Wakin,
A simple proof of the restricted isometry property for random matrices,
Constr.\ Approx.\ 28 (2008) 253--263.

\bibitem{BourgainDFKK:11}
J.\ Bourgain, S.\ Dilworth, K.\ Ford, S.\ Konyagin, D.\ Kutzarova,
Explicit constructions of RIP matrices and related problems,
Duke Math.\ J.\ 159 (2011) 145--185.

\bibitem{BourgainDFKK:11b}
J.\ Bourgain, S.\ Dilworth, K.\ Ford, S.\ Konyagin, D.\ Kutzarova,
Breaking the $k^2$ barrier for explicit RIP matrices,
STOC 2011, 637--644.

\bibitem{CaiZ:13}
T.\ T.\ Cai, A.\ Zhang,
Sharp RIP bound for sparse signal and low-rank matrix recovery,
Appl.\ Comput.\ Harmon.\ Anal.\ 35 (2013) 74--93.

\bibitem{Candes:08}
E.\ J.\ Cand\`{e}s,
The restricted isometry property and its implications for compressed sensing,
C.\ R.\ Acad.\ Sci.\ Paris, Ser.\ I 346 (2008) 589--592.

\bibitem{CandesRT:06}
E.\ J.\ Cand\`{e}s, J.\ Romberg, T.\ Tao,
Robust uncertainty principles:\ Exact signal reconstruction from highly incomplete frequency information,
IEEE Trans.\ Inform.\ Theory 52 (2006) 489--509.

\bibitem{CandesT:06}
E.\ J.\ Cand\`{e}s, T.\ Tao,
Near-optimal signal recovery from random projections:\ Universal encoding strategies?,
IEEE Trans.\ Inform.\ Theory 52 (2006) 5406--5425.

\bibitem{CheraghchiGV:13}
M.\ Cheraghchi, V.\ Guruswami, A.\ Velingker,
Restricted isometry of Fourier matrices and list decodability of random linear codes,
SIAM J.\ Comput.\ 42 (2013) 1888--1914.

\bibitem{DelapenaG:12}
V.\ de la Pe\~{n}a, E.\ Gin\'{e},
Decoupling:\ From Dependence to Independence,
Springer, 2012.

\bibitem{DeVore:07}
R.\ A.\ DeVore,
Deterministic constructions of compressed sensing matrices,
J.\ Complexity 23 (2007) 918--925.

\bibitem{Donoho:06}
D.\ L.\ Donoho,
Compressed sensing,
IEEE Trans.\ Inform.\ Theory 52 (2006) 1289--1306.

\bibitem{DonohoT:09}
D.\ Donoho, J.\ Tanner,
Observed universality of phase transitions in high-dimensional geometry, with implications for modern data analysis and signal processing,
Philos.\ Trans.\ R.\ Soc.\ A 367 (2009) 4273--4293.

\bibitem{DuarteDTLSKB:08}
M.\ F.\ Duarte, M.\ A.\ Davenport, D.\ Takhar, J.\ N.\ Laska, T.\ Sun, K.\ F.\ Kelly, R.\ G.\ Baraniuk,
Single-pixel imaging via compressive sampling,
IEEE Signal Process.\ Mag.\ 25 (2008) 83--91.

\bibitem{FickusM:15}
M.\ Fickus, D.\ G.\ Mixon,
Tables of the existence of equiangular tight frames,
arXiv:1504.00253

\bibitem{FickusMT:12}
M.\ Fickus, D.\ G.\ Mixon, J.\ C.\ Tremain,
Steiner equiangular tight frames,
Linear Algebra Appl.\ 436 (2012) 1014--1027.

\bibitem{FoucartR:13}
S.\ Foucart, H.\ Rauhut,
A Mathematical Introduction to Compressive Sensing,
Springer, 2013.

\bibitem{Goldreich:01}
O.\ Goldreich,
Foundations of cryptography I:\ Basic Tools,
Cambridge University Press, 2001.

\bibitem{HaikinZG:17}
M.\ Haikin, R.\ Zamir, M.\ Gavish,
Random subsets of structured deterministic frames have MANOVA spectra,
Proc.\ Natl.\ Acad.\ Sci.\ U.S.A.\ 114 (2017) E5024--E5033.

\bibitem{HavivR:17}
I.\ Haviv, O.\ Regev,
The restricted isometry property of subsampled Fourier matrices,
In:\ Geometric Aspects of Functional Analysis, Springer, 2017, pp.\ 163--179.

\bibitem{HugelRS:14}
M.\ H\"{u}gel, H.\ Rauhut, T.\ Strohmer,
Remote sensing via $\ell_1$-minimization,
Found.\ Comput.\ Math.\ 14 (2014) 115--150.

\bibitem{Hoeffding:63}
W.\ Hoeffding,
Probability Inequalities for Sums of Bounded Random Variables,
J.\ Amer.\ Statist.\ Assoc.\ 58 (1963) 13--30.

\bibitem{Iwen:14}
M.\ A.\ Iwen,
Compressed sensing with sparse binary matrices:\ Instance optimal error guarantees in near-optimal time,
J.\ Complexity 30 (2014) 1--15.

\bibitem{JasperMF:13}
J.\ Jasper, D.\ G.\ Mixon, M.\ Fickus,
Kirkman equiangular tight frames and codes,
IEEE Trans.\ Inform.\ Theory 60 (2013) 170--181.

\bibitem{JohnsonL:86}
W.\ B.\ Johnson, J.\ Lindenstrauss,
Extensions of Lipschitz mappings into a Hilbert space,
Contemp.\ Math.\ 26 (1984) 189--206.

\bibitem{KrahmerMR:14}
F.\ Krahmer, S.\ Mendelson, H.\ Rauhut,
Suprema of chaos processes and the restricted isometry property,
Comm.\ Pure Appl.\ Math.\ 67 (2014) 1877--1904.

\bibitem{MagsinoMP:19}
M.\ Magsino, D.\ G.\ Mixon, H.\ Parshall,
Kesten--McKay law for random subensembles of Paley equiangular tight frames,
arXiv:1905.04360

\bibitem{MendelsonPT:09}
S.\ Mendelson, A.\ Pajor, N.\ Tomczak-Jaegermann,
Uniform uncertainty principle for Bernoulli and subgaussian ensembles,
Constr.\ Approx.\ 28 (2009) 277--289.

\bibitem{Mixon:15}
D.\ G.\ Mixon,
Explicit matrices with the restricted isometry property:\ Breaking the square-root bottleneck.
In:\ Compressed sensing and its applications, Birkh\"{a}user, 2015, pp.\ 389--417.

\bibitem{MonajemiJGD:13}
H.\ Monajemi, S.\ Jafarpour, M.\ Gavish, D.\ L.\ Donoho, Stat 330/CME 362 Collaboration,
Deterministic matrices matching the compressed sensing phase transitions of Gaussian random matrices,
Proc.\ Natl.\ Acad.\ Sci.\ U.S.A.\ 110 (2013) 1181--1186.

\bibitem{Renes:07}
J.\ M.\ Renes,
Equiangular tight frames from Paley tournaments,
Linear Algebra Appl.\ 426 (2007) 497--501.

\bibitem{RudelsonV:08}
M.\ Rudelson, R.\ Vershynin,
On sparse reconstruction from Fourier and Gaussian measurements,
Comm.\ Pure Appl.\ Math.\ 61 (2008) 1025--1045.

\bibitem{Tao:07}
T.\ Tao,
Open question:\ deterministic UUP matrices,
\url{https://terrytao.wordpress.com/2007/07/02/open-question-deterministic-uup-matrices/}

\bibitem{Tao:12}
T.\ Tao,
Topics in random matrix theory,
Amer.\ Math.\ Soc., 2012.

\bibitem{TaubmanM:12}
D.\ Taubman, M.\ Marcellin,
JPEG2000 image compression fundamentals, standards and practice,
Springer, 2012.

\bibitem{TillmannP:13}
A.\ M.\ Tillmann, M.\ E.\ Pfetsch,
The computational complexity of the restricted isometry property, the nullspace property, and related concepts in compressed sensing,
IEEE Trans.\ Inform.\ Theory 60 (2013) 1248--1259.

\bibitem{Tropp:15}
J.\ A.\ Tropp,
An introduction to matrix concentration inequalities,
Found.\ Trends Mach.\ Learn.\ 8 (2015) 1--230.

\bibitem{Tropp:08}
J.\ A.\ Tropp,
On the conditioning of random subdictionaries,
Appl.\ Comput.\ Harmon.\ Anal.\ 25 (2008) 1--24.

\bibitem{Tropp:12}
J.\ A.\ Tropp,
User-friendly tail bounds for sums of random matrices,
Found.\ Comput.\ Math.\ 12 (2012) 389--434.

\bibitem{Vershynin:18}
R.\ Vershynin,
High-dimensional probability:\ An introduction with applications in data science,
Cambridge University Press, 2018.

\bibitem{WangBP:16}
T.\ Wang, Q.\ Berthet, Y.\ Plan,
Average-case hardness of RIP certification,
NeurIPS (2016) 3819--3827.

\end{thebibliography}
\end{document}